\documentclass[12pt]{article}
\usepackage{fullpage}
\usepackage{amsthm}
\usepackage{amsmath}
\usepackage{amssymb}
\usepackage{mathtools}
\usepackage{enumerate}
\usepackage{thm-restate}
\usepackage{bibspacing}

\usepackage[margin=1in]{geometry}
\begin{document}
\newtheorem{theorem}{Theorem}[section]
\newtheorem{lemma}[theorem]{Lemma}
\newtheorem{corollary}[theorem]{Corollary}
\newtheorem{fact}[theorem]{Fact}
\newtheorem{claim}{Claim}
\newtheorem*{claim*}{Claim}
\newtheorem*{conjecture*}{Conjecture}
\theoremstyle{definition}
\newtheorem{definition}[theorem]{Definition}
\newtheorem{assumption}[theorem]{Assumption}

\newcommand{\bits}{\{0,1\}}
\newcommand{\poly}{\mathrm{poly}}

\newcommand{\evnote}[1]{\marginpar{#1}}
\newcommand{\evnotelong}[1]{\textbf{Note:} #1 \textbf{End note.}}

\newcommand{\set}[1]{\left \{{#1} \right \}}
\newcommand{\rob}[1]{\left( #1 \right)}
\newcommand{\sqb}[1]{\left[ #1 \right]}
\newcommand{\cub}[1]{\left\{ #1 \right\} }
\newcommand{\rb}[1]{\left( #1 \right)} 
\newcommand{\abs}[1]{\left| #1 \right|} 
\newcommand{\zo}{\{0, 1\}}
\newcommand{\zonzo}{\zo^n \to \zo}
\newcommand{\zon}{\zo^n}
\newcommand{\e}{\epsilon}
\newcommand{\eps}{\epsilon}
\newcommand{\nat}{\mathbb{N}}

\newcommand{\acz}{\mathrm{AC}^0}
\newcommand{\accz}{\mathrm{ACC}^0}
\newcommand{\nc}{\mathrm{NC}}
\newcommand{\tcz}{\mathrm{TC}^0}
\newcommand{\aczp}{\mathrm{AC}^0[\oplus]}
\newcommand{\ncz}{\mathrm{NC}^0}

\newcommand{\conj}{\textsf{Conj}}
\newcommand{\comm}{\textsf{Comm}}
\newcommand{\id}{\textsf{id}}
\newcommand{\comab}{\alpha\beta\alpha^{-1}\beta^{-1}}
\newcommand{\chat}{\widehat{C}}
\newcommand{\gtot}{G^{\, t}}

\newcommand*{\TitleFont}{%
      \fontsize{19}{15}%
      \selectfont}


\title{\TitleFont Iterated group products and leakage resilience against NC$^1$}
\author{Eric Miles\let\thefootnote\relax\thanks{Supported by NSF grants CCF-0845003 and CCF-1319206. Email: \texttt{enmiles@ccs.neu.edu}.}}
\date{December 3, 2013}
\maketitle

\begin{abstract}
We show that if NC$^1 \neq$ L, then for every element $\alpha$ of the alternating group $A_t$, circuits of depth $O(\log t)$ cannot distinguish between a uniform vector over $(A_t)^t$ with product $= \alpha$ and one with product $=$ identity. Combined with a recent construction by the author and Viola in the setting of leakage-resilient cryptography [STOC '13], this gives a compiler that produces circuits withstanding leakage from NC$^1$ (assuming NC$^1 \neq$ L). For context, leakage from NC$^1$ breaks nearly all previous constructions, and security against leakage from P is impossible.

We build on work by Cook and McKenzie [J.\ Algorithms '87] establishing the relationship between L $=$ logarithmic space and the symmetric group $S_t$. Our techniques include a novel algorithmic use of commutators to manipulate the cycle structure of permutations in $A_t$.
\end{abstract}


\section{Introduction}
The interplay between {\em group theory} and {\em computational complexity} has been the source of a number of elegant constructions and computational insights.
A line of work \cite{Barrington89, CookM87, BarringtonT88, Ben-OrC92, CaiL94, Cleve91,ImmermanL95} in the late 1980s gave characterizations of various complexity classes in terms of products over finite groups. A primary motivation of some of this work was to obtain an efficient simulation of circuits by branching programs, such as in the celebrated theorem of Barrington \cite{Barrington89}. 
Beyond this however, the underlying encoding of computation by group products has proven to be a useful tool in other areas.

A typical encoding has the following form. Let $G$ be a group with identity element $\id$, let $f$ be a Boolean function, and let $x$ be an input to $f$. Then the encoding $E(f,x)$ is a vector $(g_1,\ldots,g_t) \in \gtot$ such that $\prod_i g_i = \id \Leftrightarrow f(x) = 0$. For such encodings a number of efficiency aspects are of interest, including the size of the group, the length of the encoding, and the complexity of computing it. In some cases a short encoding can be computed with much fewer resources than are needed to compute $f$, yielding the result that the decoding problem, i.e.\ computing iterated products over $G$, is as hard as computing $f$.

Recently, the author and Viola \cite{MilesV-leak} gave an application of these group encodings to {\em leakage-resilient cryptography}. This area studies cryptographic models in which the adversary obtains more information from an algorithm than just its input/output behavior. The extra information is commonly modeled by providing the adversary with the output of a computationally restricted ``leakage function'' of its choosing, applied to (the bits carried on) the wires of a circuit implementing the algorithm.
A general goal in this area is to compile any circuit $C$ into a functionally equivalent circuit $\chat$ such that any attack on $\chat$ exploiting this extra information can in fact be carried out just using input/output access, and hence does not succeed under standard hardness assumptions. 



The \cite{MilesV-leak} construction provides a rather generic way to construct such a compiler, given any group $G$ satisfying the following property with respect to a set of leakage functions $\mathcal{L}$. This property says that for an element $\alpha \in G$, functions in $\mathcal{L}$ cannot distinguish between vectors with product $= \alpha$ and those with product $= \id$. 

\begin{definition}
\label{hard-def}
Let $G$ be a group. For $\alpha \in G$ and $t \in \mathbb{N}$, the {\em $\alpha$-product problem over $\gtot$} is to decide, given $(x_1,\ldots,x_t) \in \gtot$ such that $\prod_i x_i \in \{\alpha,\id\}$, which product it has.

Let $D_\alpha$ denote the uniform distribution over $\{(x_1,\ldots,x_t) \in \gtot\ |\ \prod_i x_i = \alpha\}$. A set of functions $\mathcal{L}$ is {\em $(\eps,\alpha)$-fooled by $\gtot$} if $\Delta(\ell(D_\alpha),\ell(D_\id)) \leq \eps$ for every $\ell \in \mathcal{L}$. 
\end{definition}	

The \cite{MilesV-leak} construction obtains security against leakage classes that are $(\eps,\alpha)$-fooled by $\gtot$ for {\em every} $\alpha \in G$. A key technical contribution there is to show that a number of well-studied classes of functions have this property when $G = A_5$, where $A_u$ denotes the group of even permutations on $u$ points.
In particular they show that for every $\alpha \in A_5$, functions including parity, majority, inner product, and communication protocols are $(t^{-\omega(1)},\alpha)$-fooled by $(A_5)^t$, yielding a compiler that provides security against leakage from these functions. 

They also make the following conjecture, whose proof would yield a compiler that provides security against leakage from the class NC$^1$ of log-depth, fan-in-2 circuits.

\begin{conjecture*}[\cite{MilesV-leak}]
If {\em NC}$^1 \neq$ {\em L} then for every $\alpha \in A_t$, {\em NC}$^1$ is $(t^{-\omega(1)},\alpha)$-fooled by $(A_t)^t$.
\end{conjecture*}

Some support for this conjecture comes from the work of Cook and McKenzie \cite{CookM87} on the relationship between L = log-space and iterated group products. Their results can be used to show that the following problem is L-complete: given $x = (x_1,\ldots,x_t) \in (A_t)^t$, determine if $\prod_i x_i = \id$. (The difference from Definition \ref{hard-def} is that $x$'s product is not guaranteed to be in $\{\alpha,\id\}$.) If one could instead show that the $\alpha$-product problem is L-complete for every $\alpha \in A_t$, the conjecture follows from the random self-reducibility of iterated group products \cite{Kilian88}. However, it was not clear how to show this for even a single $\alpha$.

\subsection{Our results}
We show that indeed the $\alpha$-product problem is L-complete for every $\alpha \in A_t$, and as a result we prove the above conjecture.

\begin{restatable}{theorem}{thmallalphahard} \label{thm:allalphahard}
Assume that there is a circuit family $C$ of depth $O(\log t)$ such that for sufficiently large $t$, there exists $\alpha \in A_t$ such that $C$ decides the $\alpha$-product problem over $(A_t)^t$. Then {\em NC}$^1 =$ \em{L}.
\end{restatable}

\begin{restatable}{corollary}{thmmain}
\label{thm:nc1}
If {\em NC}$^1 \neq $ {\em L}, then {\em NC}$^1$ is $(t^{-\omega(1)},\alpha)$-fooled by $(A_t)^t$ for infinitely many $t$ and all $\alpha \in A_t$.

More precisely, if {\em NC}$^1 \neq$ {\em L} then for all $k$, infinitely many $t$, and all $\alpha \in A_t$, the class of {\em NC}$^1$ circuits with depth $k \log t$ and output length $k \log t$ is $(t^{-k},\alpha)$-fooled by $(A_t)^t$.
\end{restatable}

In combination with \cite{MilesV-leak}, Corollary \ref{thm:nc1} yields a compiler that is secure against NC$^1$ leakage. As noted in \cite{MilesV-leak}, even 1 bit of leakage from NC$^1$ breaks nearly all previous constructions \cite{IshaiSW03,FaustRRTV10,JumaV10,DziembowskiF12,GoldwasserR12, Rothblum12} (in fact leakage from TC$^0 \subseteq$ NC$^1$ is already enough to break these). Furthermore, securing circuits against arbitrary polynomial-time leakage is known to be {\em impossible}, due to the impossibility of obfuscation \cite{BarakGIRSVY01}.
We note that, like some other compilers \cite{FaustRRTV10,GoldwasserR10,JumaV10,DziembowskiF12}, in the multi-query setting the \cite{MilesV-leak} construction uses a so-called ``secure hardware component''. 
Specifically, the compiled circuits have a gate that produces a uniform sample from $(A_t)^t$ with product $= \id$, and whose internal computation is not visible to the leakage function. 

\paragraph{On the amount of leakage.}
The amount of NC$^1$ leakage tolerated using Corollary \ref{thm:nc1} is logarithmic in the security parameter, which is smaller than one might wish (and smaller than can be achieved against other leakage classes). This limitation seems to be inherent when starting from an assumption such as NC$^1$ $\neq$ L. One approach to circumventing this is by instead starting from a {\em compression bound}, i.e.\ a bound on the correlation between NC$^1$ circuits with $> 1$ output bit and functions in L. Indeed, the compression bound due to Dubrov and Ishai \cite{DubrovI06} against AC$^0$ was used by \cite{FaustRRTV10,MilesV-leak} to achieve security against a linear amount of AC$^0$ leakage. 

In this setting however, there is an issue which is that the proof of Theorem \ref{thm:allalphahard} gives a Turing reduction to the $\alpha$-product problem, as opposed to a many-one reduction. With such a reduction it is not clear how to translate a (hypothetical) NC$^1$ circuit compressing the $\alpha$-product problem into a circuit compressing an arbitrary language in L, and therefore a generic assumption on the inability of NC$^1$ to compress L does not immediately yield improved leakage bounds. In light of this, it would be interesting to show that the $\alpha$-product problem is L-complete under many-one reductions. Jumping ahead, showing this for any fixed $\alpha \in A_t$ would suffice to show it for all $\alpha$ via Theorem \ref{thm:localmap}.

\medskip

We refer the reader to \cite{MilesV-leak} for further details on leakage-resilience, and now give an overview of the proof of Theorem \ref{thm:allalphahard}. The proof of Corollary \ref{thm:nc1} given Theorem \ref{thm:allalphahard} uses straightforward techniques and appears in \S \ref{sec:puttingtogether}.

\subsection{Techniques}\label{sec:techniques}

\paragraph{The Cook-McKenzie construction.}
Our starting point is the work by Cook and McKenzie \cite{CookM87} who show that a number of problems related to $S_t$, the group of all permutations on $t$ points, are complete for L. A key tool in their work is the following construction of a permutation that encodes acceptance of a given branching program on a given input. (We equate L with the set of polynomial-size branching programs.)

\begin{restatable}[\cite{CookM87}]{theorem}{cookmckenzie} \label{thm:cook-mckenzie}
There exists $t = O(s)$ and a circuit $C$ of depth $O(\log s)$ such that, on input $(B,x)$ where $B$ is a branching program of size $s \geq |x|$, $C$ outputs a permutation $\sigma \in S_{t}$ such that $B$ accepts $x$ iff $1$ and $t$ are in the same cycle in $\sigma$'s disjoint cycle representation.
\end{restatable}

Recall that any permutation can be uniquely written as a product of disjoint cycles. Throughout, when a permutation $\sigma \in S_t$ is the input or output of an algorithm, it is given in $(t \log t)$-bit pointwise representation as $(\sigma(1),\ldots,\sigma(t))$. Inversion and pairwise multiplication in $S_t$ can then be implemented in NC$^1$, because each essentially amounts to indexing an element from an array of length $t$.

Theorem \ref{thm:cook-mckenzie} is proved by constructing a permutation $\sigma$ that performs one step of a depth-first search in $B$ when projecting to the edges consistent with the input $x$. The nodes are labeled by $[t]$, and the labels of the start and accept nodes ($1$ and $t$ wlog) are in the same cycle of $\sigma$ iff the accept node is reachable from the start node. 

As observed to us by Eric Allender and V. Arvind (personal communication), Theorem \ref{thm:cook-mckenzie} can be used to show that the following very natural problem is L-complete: given $(x_1,\ldots,x_t) \in (S_t)^t$, decide if $\prod_i x_i = \id$. A simple embedding trick shows that this problem remains L-complete even over the group $A_t$ (see \S\ref{sec:decide-id-hard}). 

However it was not clear whether Theorem \ref{thm:cook-mckenzie} could be used to show that the $\alpha$-product problem is L-complete, for every or even for a single $\alpha \in A_t$. The issue is that the permutation $\sigma$ constructed in Theorem \ref{thm:cook-mckenzie} depends on the structure of the branching program $B$ when projecting to the edges whose labels match the input $x$. So while the reduction to the ``product $\stackrel{?}{=} \id$'' problem always produces a vector with product $= \id$ when $B(x) = 0$, when $B(x) = 1$ the product can change depending on $x$, rather than being fixed to some $\alpha$.

Before explaining how we overcome this, we briefly review some facts and terminology about permutations that will be used; a more detailed discussion appears in e.g.\ \cite[\S 2]{Wilson09}.

\paragraph{Permutation groups.}
$S_t$ is the group of all permutations on $[t] = \{1,\ldots,t\}$. Any permutation can be written as a product of transpositions, and the permutation is called {\em even} or {\em odd} depending on whether the number of transpositions in this product is even or odd. Thus the product of two permutations is even iff both are even or both are odd. $A_t \subset S_t$ is the group of all even permutations on $[t]$.

A {\em $k$-cycle} $(a_1\ \cdots\ a_k)$ is a permutation that maps $a_i \mapsto a_{i+1}$ for $i < k$ and $a_k \mapsto a_1$. A $k$-cycle is an even permutation iff $k$ is odd, because any $k$-cycle can be written as a product of $k-1$ transpositions $(a_1\ a_2)(a_1\ a_3)\cdots(a_1\ a_k)$. Every permutation can be uniquely written as a product of disjoint cycles, and the list of these cycles' lengths is the {\em cycle type} of the permutation.

$\alpha,\beta \in G$ are {\em conjugate} if there exists $\gamma \in G$ such that $\gamma^{-1} \alpha \gamma = \beta$. 
Conjugacy is an equivalence relation and partitions $G$ into its conjugacy classes. In $S_t$, two permutations are conjugate iff they have the same cycle type. In $A_t$ this continues to hold if that cycle type contains either an even-length cycle or two cycles of the same length, but for cycle types consisting of distinct odd lengths (including length $1$) there are two distinct conjugacy classes. For example, the two 5-cycles $(1\ 2\ 3\ 4\ 5)$ and $(1\ 2\ 3\ 5\ 4)$ are conjugate in $S_5$ and $S_6$ but not in $A_5$ nor $A_6$.

The {\em commutator} of two elements $\alpha,\beta \in G$ is denoted $[\alpha,\beta]$ and defined by $[\alpha,\beta] := \alpha \beta \alpha^{-1} \beta^{-1}$. We will use the fact that for every $\alpha$, $[\alpha,\id] = [\id,\alpha] = \id$. 

\medskip

We now explain the proof of Theorem \ref{thm:allalphahard}, which has two steps. 
We first show that if NC$^1$ can decide the $\alpha$-product problem for any fixed $\alpha \in A_t$, then it can do so for every $\alpha \in A_t$. Next we show that the $\alpha$-product problem is L-complete for the specific choice of $\alpha = (1\ 2)(3\ 4)$. The combination of these implies that if NC$^1$ can decide the $\alpha$-product problem for any $\alpha$, then NC$^1 =$ L. In light of the above discussion on the amount of leakage, we remark that only the latter step uses a Turing reduction.

\paragraph{Mapping group products.}
We reduce the $\alpha$-product problem to the $\beta$-product problem, for any $\id \neq \alpha,\beta \in A_t$, by constructing an NC$^1$ map $$f_{\alpha \to \beta} : (A_t)^m \to (A_t)^{O(tm)}$$ that maps $\alpha$-products to $\beta$-products while preserving $\id$-products. Namely, this map satisfies 
\begin{equation} \label{eq:a2bmap}
\prod_i x_i = \alpha\ \Rightarrow\ \prod_i f_{\alpha \to \beta}(x)_i = \beta \qquad\mbox{and}\qquad \prod_i x_i = \id\ \Rightarrow\ \prod_i f_{\alpha \to \beta}(x)_i = \id
\end{equation}
for every $x = (x_1,\ldots,x_m) \in (A_t)^m$. (To prove Theorem \ref{thm:nc1} we choose $m = t$, but here we keep them separate to make the blowup in output length clear.) Moreover $f_{\alpha \to \beta}$ has the nice property that each output element depends on only 1 input element, i.e.\ it is {\em 1-local}.

If $\alpha$ and $\beta$ are in the same conjugacy class, then there is a simple 1-local NC$^1$ map with no blowup in length that satisfies (\ref{eq:a2bmap}). This map is $$(x_1,\ \ldots,\ x_m)\ \mapsto\ (\gamma^{-1} x_1,\ \ldots,\ x_m \gamma)$$ where $\gamma \in A_t$ satisfies $\gamma^{-1} \alpha \gamma = \beta$. But because $\alpha$ and $\beta$ are in the same conjugacy class only if they have the same cycle type, to map between permutations of different cycle types a different technique is needed.

Our idea is to use {\em commutation} to map between permutations of different cycle types, and combined with conjugation as above this allows us to construct a map from any $\alpha$ to any $\beta$. As noted previously commutation is identity-preserving, i.e.\ $[\gamma,\id] = [\id,\gamma] = \id$ for every $\gamma$, which makes it a good candidate for maps satisfying (\ref{eq:a2bmap}). 
One limitation is that commutation doubles the length of the vector, i.e.\ we compute $$(x_1,\ \ldots,\ x_m)\ \mapsto\ (x_1,\ \ldots,\ x_m\gamma,\ x_m^{-1},\ \ldots,\ x_1^{-1}\gamma^{-1})$$ to map $\alpha$-products to $[\alpha,\gamma]$-products while preserving $\id$-products. But by using $\leq \log t + O(1)$ commutations, here the total length increase is limited to a factor of $O(t)$.

We note that our use of commutation is different than in Barrington's construction \cite{Barrington89}. There it is used to simulate AND, and each commutator is formed from two vectors whose products are both unknown. Here we use it to manipulate cycle structure, and each commutator is formed from a vector whose product is unknown and a fixed group element.

Our framework for using commutation to map between different cycle types has three steps.
\begin{enumerate}
\item {\em Reduce to a product of two disjoint transpositions.} (Lemma \ref{lem:to-doub-transp})

We first convert any $\alpha \neq \id$ into a product of two disjoint transpositions. For example if $\alpha$'s disjoint cycle representation contains a 4-cycle $(a\ b\ c\ d)$, then commutating with $\gamma := (a\ b)(c\ d)$ yields $[\alpha,\gamma] = (a\ c)(b\ d)$.

\item {\em Grow the number of transpositions.} (Lemma \ref{lem:growtranspositions})

From a product of disjoint transpositions, we use one commutation to double the number of transpositions, and $\log k$ commutations to convert from $2$ to $2k$ disjoint transpositions. For example if $\alpha = (a\ b)(c\ d)$, then commutating with $\gamma := (a\ e)(b\ f)(c\ g)(d\ h)$ yields $[\alpha,\gamma] = (a\ b)(c\ d)(e\ f)(g\ h)$.

\item {\em Combine transpositions into cycles.} (Lemmas \ref{lem:oddcycles}-\ref{lem:evencycles})

We finally convert products of disjoint transpositions into longer cycles. For example if $\alpha = (a_1\ b_1)\cdots(a_k\ b_k)$, then commutating with $\gamma := (a_1\ b_1\ a_2\ b_2\ \cdots\ a_k\ b_k\ c)$ yields the $(2k+1)$-cycle $[\alpha,\gamma] = (a_1\ \cdots\ a_k\ b_k\ \cdots\ b_1\ c)$.
\end{enumerate}

In total we use $\leq \log t + O(1)$ commutations. This turns out to be tight for certain starting and target permutations as shown in the following theorem, and thus for these permutations a map satisfying (\ref{eq:a2bmap}) with smaller output length requires different techniques. (The theorem also holds if both commutations and conjugations are allowed.)

\begin{theorem}
There exist $\alpha,\beta \in A_t$ such that $\alpha$ cannot be converted to $\beta$ with fewer than $\log(t) - 1$ commutations. That is, for every $\ell < \log(t) - 1$ and every sequence $\gamma_1,\ldots,\gamma_\ell \in A_t$, $\beta \neq [[\cdots[ [\alpha,\gamma_1], \gamma_2], \cdots], \gamma_\ell]$.
\end{theorem}
\begin{proof}
For $\alpha \in A_t$, let $M(\alpha) := |\{i \in [t]\ |\ \alpha(i) \neq i\}|$ denote the number of points moved (i.e.\ not fixed) by $\alpha$. We show that $M([\alpha,\gamma]) \leq 2 \cdot M(\alpha)$ for every $\alpha$ and $\gamma$, i.e.\ the number of points moved by $[\alpha,\gamma]$ is at most twice the number of points moved by $\alpha$. This implies the theorem by choosing $\alpha = (1\ 2\ 3)$ and $\beta = (1\ 2\ \cdots\ t)$, because $M(\alpha) = 3$ and $M(\beta) = t$.

Pick $\gamma \in A_t$. Observe that $\gamma\alpha^{-1}\gamma^{-1}$ has the same cycle type as $\alpha$ because it is conjugate to $\alpha$ in $S_t$. This implies $M(\gamma \alpha^{-1} \gamma^{-1}) = M(\alpha)$, and therefore $[\alpha,\gamma]=\alpha\gamma\alpha^{-1}\gamma^{-1}$ moves at most $2 \cdot M(\alpha)$ points because any such point must be moved by either $\alpha$ or $\gamma\alpha^{-1}\gamma^{-1}$.
\end{proof}

\paragraph{Hardness for a single element.}
To prove that the $(1\ 2)(3\ 4)$-product problem is L-complete, we reduce from the problem of deciding if $x \in (A_t)^t$ has product $= \id$. That is, from any $x \in (A_t)^t$ we compute a vector $y \in (A_t)^t$ such that $y$ has product $= \id$ if $x$ does, and otherwise $y$ has product $= (1\ 2)(3\ 4)$. 

As in step 1 above, we first show that if $\prod_i x_i = \alpha \neq \id$ then there is some $\gamma_1 \in A_t$ (which depends on $\alpha$) such that $[\alpha,\gamma_1]$ is a double-transposition, i.e.\ a product of two disjoint transpositions. Then because all double-transpositions are conjugate in $A_t$ there is some $\gamma_2 \in A_t$ such that $\gamma_2^{-1} \cdot [\alpha,\gamma_1] \cdot \gamma_2 = (1\ 2)(3\ 4)$, and a vector with this product can be computed from $x$ in NC$^1$ if we know $\gamma_1,\gamma_2$.

The problem of course is that we do not know $\gamma_1,\gamma_2$ without knowing $\alpha = \prod_i x_i$ which we cannot compute. We resolve this by showing that for any $\alpha \neq \id$, $\gamma_1$ and $\gamma_2$ as above can be taken from the set of permutations that are fixed on all but $\leq 8$ points in $[t]$. Since there are $< {t \choose 8} \cdot |A_8| = t^{O(1)}$ such permutations, we can thus construct a set of $t^{O(1)}$ vectors such that if $x$ has product $= \id$ then they all do, and otherwise some vector has product $= (1\ 2)(3\ 4)$. Then the proof is completed by applying an NC$^1$ circuit deciding the $(1\ 2)(3\ 4)$-product problem to each vector and a depth-$O(\log t)$ OR tree to the results.

\paragraph{Organization.}
The rest of the paper is organized as follows. In \S \ref{sec:mapping} we describe our mapping between group products and construct the function $f_{\alpha \to \beta}$ defined above for any $\alpha$ and $\beta$. In \S \ref{sec:1234} we show that the $(1\ 2)(3\ 4)$-product problem is L-complete. Together these prove Theorem \ref{thm:allalphahard}, and in \S \ref{sec:puttingtogether} we prove Corollary \ref{thm:nc1}.

\section{Mapping group products} \label{sec:mapping}

In this section we prove the following theorem.

\begin{theorem} \label{thm:localmap}
Let $t \equiv 2\, (\bmod\ 4)$ and let $\id \neq \alpha, \beta \in A_t$. Then for all $m$ there is a 1-local function $f : (A_t)^m \to (A_t)^{O(tm)}$ computable in depth $O(\log t)$ that maps $\alpha$-products to $\beta$-products while preserving the identity, i.e.\ that satisfies $\forall x = (x_1,\ldots,x_m) \in (A_t)^m$:
$$\prod_i x_i = \alpha \Rightarrow \prod_i f(x)_i = \beta \qquad\mbox{and}\qquad \prod_i x_i = \id \Rightarrow \prod_i f(x)_i = \id.$$
\end{theorem}

The function $f$ is constructed by concatenating compositions of functions from the following two families, where the compositions are given by Lemma \ref{lem:convert}.
$$\conj = \{\conj_\gamma(\alpha) := \gamma^{-1} \alpha  \gamma\ |\ \gamma \in A_t\} \quad\qquad \comm = \{\comm_\gamma(\alpha) := \alpha\gamma\alpha^{-1}\gamma^{-1}\ |\ \gamma \in A_t\}$$

\begin{lemma} \label{lem:convert}
Let $t \equiv 2\, (\bmod\ 4)$ and let $\alpha, \beta \in A_t$ such that $\alpha \neq \id$ and $\beta$ is either a cycle of odd length $k$ or is the product of two disjoint even-length cycles of total length $k$. Then, there is a sequence $f_1, \ldots, f_\ell \in (\conj\, \cup \comm)$ such that $f(\alpha) = \beta$, where $f := f_\ell \circ \cdots \circ f_1 $ and $\ell = \log k + O(1)$.
\end{lemma}

Any function given by this lemma yields a 1-local function $f : (A_t)^m \to (A_t)^{m \cdot 2^\ell}$ computable in depth $O(\log t)$ that maps $\alpha$-products to $\beta$-products while preserving the identity. 

\begin{proof}[Proof of Theorem \ref{thm:localmap}]
We prove the case $m = 1$, but the argument extends immediately to any $m$. Fix $\alpha,\beta \in A_t$, and consider the unique representation of $\beta$ as a set of disjoint cycles $C = \{\sigma_1,\ldots,\sigma_s\} \subset S_t$. (Here $C$ contains only those cycles with length $> 1$.) The idea is to apply Lemma \ref{lem:convert} to each cycle $\sigma \in C$, obtaining $f' : A_t \to (A_t)^{O(|\sigma|)}$ such that $\prod_i f'(\alpha)_i = \sigma$. Then letting $f$ output the concatenation of these $f'$, the resulting function maps $\alpha$-products to $\beta$-products while preserving the identity. Further, its output length is $\sum_{\sigma \in C} O(|\sigma|) = O(t)$.

The only technical complication has to do with cycles of even length: if $\sigma$ is a cycle of even length then it is an odd permutation, and so there can be no composition of functions from \conj\ and \comm\ that maps $\alpha \mapsto \sigma$.  We handle this by pairing the cycles of even length, which we can do because $C$ must contain an even number of cycles of even length as each is an odd permutation and $\beta$ is an even permutation. So, for each such pair of even-length cycles $\sigma,\sigma' \in C$, we instead apply Lemma \ref{lem:convert} to get a function $f'$ such that $\prod_i f'(\alpha)_i = \sigma \cdot \sigma'$.
\end{proof}

\subsection{Proof of Lemma \ref{lem:convert}}
To prove Lemma \ref{lem:convert}, we implement the procedure described in \S\ref{sec:techniques}. Namely, we first use $\leq 2$ commutations to convert a given $\alpha \in A_t$ to a double-transposition, then $\log k$ commutations to convert to a product of roughly $k/2$ disjoint transpositions, and finally $\leq 2$ commutations and 1 conjugation to convert to either a cycle of odd length $k$ or to the product of two disjoint even-length cycles with total length $k$.

\begin{lemma} \label{lem:to-doub-transp}
Let $t \geq 4$. For every $\id \neq \alpha \in A_t$, there exist $\gamma_1,\gamma_2 \in A_t$ such that $[[\alpha,\gamma_1],\gamma_2]$ is a double-transposition. Further, each $\gamma_i$ is either a double-transposition or a 3-cycle.
\end{lemma}
\begin{proof}
We consider five cases based on $\alpha$'s cycle structure. In all cases except the last, in fact only one commutation is needed to obtain a double-transposition, but we can use two by noting that $[(a\ b)(c\ d), (a\ b\ c)] = (a\ d)(b\ c)$.
\begin{enumerate}
\item If $\alpha$ contains a double-transposition $(a\ b)(c\ d)$, then $[\alpha, (a\ b\ c)] = (a\ d)(b\ c)$.

\item If $\alpha$ is a 3-cycle $(a\ b\ c)$, then $[\alpha, (a\ b)(c\ d)] = (a\ d)(b\ c)$.

\item If $\alpha$ contains two 3-cycles $(a\ b\ c)(d\ e\ f)$, then $[\alpha, (a\ d)(c\ f)] = (a\ d)(b\ e)$.

\item If $\alpha$ contains a 4-cycle $(a\ b\ c\ d)$, then $[\alpha, (a\ b)(c\ d)] = (a\ c)(b\ d)$.

\item If $\alpha$ contains a $(k \geq 5)$-cycle $(a_1\ \cdots\ a_k)$, then $[\alpha, (a_2\ a_3\ a_4)] = (a_1\ a_4\ a_3)$ and apply case 2.
\end{enumerate}
\end{proof}

\begin{lemma} \label{lem:growtranspositions}
For every double-transposition $\alpha \in A_t$ and even $k \leq t/2$, there exist $\gamma_1,\ldots,\gamma_{\log k} \in A_t$ such that $[[\cdots[[\alpha,\gamma_1], \gamma_2], \cdots], \gamma_{\log k}]$ is a product of $k$ disjoint transpositions.
\end{lemma}
\begin{proof}
Given $\alpha = (a\ b)(c\ d)$, we can double the number of transpositions by commutating with $\gamma = (a\ e)(b\ f)(c\ g)(d\ h)$ to get $[\alpha,\gamma] =(a\ b)(c\ d)(e\ f)(g\ h)$.  Repeating this $\log(k) - 1$ times (with appropriate modifications to $\gamma$) grows the number of transpositions from 2 to $k$. To handle $k$ that is not a power of 2, note that any $(a\ b)(c\ d)$ can be ``maintained'' rather than doubled by instead commutating with $(a\ b\ c)$ as in the proof of Lemma \ref{lem:to-doub-transp}.
\end{proof}

We now show how to commutate a product of disjoint transpositions to obtain either an odd-length cycle (Lemma \ref{lem:oddcycles}) or the product of two even-length cycles (Lemma \ref{lem:evencycles}). 

\begin{lemma} \label{lem:oddcycles}
Let $t$ be even and $\beta \in A_t$ be any cycle of odd length $5 \leq k \leq t-1$. For any $\alpha \in A_t$ that is the product of either $(k-1)/2$ or $(k-3)/2$ disjoint transpositions (depending on which is even), there are $\gamma_1,\gamma_2,\gamma_3 \in A_t$ such that either $[\gamma_1^{-1}\alpha\gamma_1,\gamma_2] = \beta$ or $[[\gamma_1^{-1}\alpha\gamma_1,\gamma_2],\gamma_3] = \beta$.
\end{lemma}
\begin{proof}
We first show that $\alpha$ can be converted to a $k$-cycle with $\leq 2$ commutations. Afterwards we observe that by first using 1 conjugation, these commutations can be made to produce the specific $k$-cycle $\beta$.

If $(k-1)/2$ is even, then let $\alpha := (a_1\ b_1)\cdots(a_{k'}\ b_{k'})$ be any product of $k' := (k-1)/2$ disjoint transpositions. Choosing $\gamma := (a_1\ b_1\ a_2\ b_2\ \cdots\ a_{k'}\ b_{k'}\ c) \in A_t$, where $c$ is distinct from all $k-1$ points permuted by $\alpha$, we get that
$$[\alpha,\gamma] = (a_1\ \cdots\ a_{k'}\ b_{k'}\ \cdots\ b_1\ c)$$
is a $k$-cycle. (We only need one commutation in this case.)

If instead $(k-3)/2$ is even (so $k \geq 7$), then let $\alpha$ be any product of $(k-3)/2$ disjoint transpositions. First, commutate once as above to get a $(k-2)$-cycle $$\mu := (a_1\ \cdots\ a_{k-2}).$$ We now show that there is another $(k-2)$-cycle $\pi \in A_t$ such that $\mu \pi$ is a $k$-cycle. This implies that we can convert $\mu$ to a $k$-cycle with one more commutation, namely by commutating with $\gamma \in A_t$ such that $\gamma \mu^{-1} \gamma^{-1} = \pi$. (Such $\gamma$ must exist because the set of $(k-2)$-cycles forms a conjugacy class in $A_t$ when $t \geq k+1$.) Take the $(k-2)$-cycle $$\pi := (a_1\ c_1\ c_2\ c_3\ a_3\ a_4\ \cdots\ a_{k-5}\ a_2)$$ where $c_1,c_2,c_3$ are distinct from the $k-2 \leq t-3$ points permuted by $\mu$. Then letting $a_i \stackrel{\mathrm{by}\ 2}{\cdots} a_j$ denote the sequence $a_i\ a_{i+2}\ a_{i+4}\ \cdots\ a_j$, we have that $$\mu \pi = (a_2 \stackrel{\mathrm{by}\ 2}{\cdots} a_{k-5}\ a_{k-4}\ a_{k-3}\ a_{k-2}\ c_1\ c_2\ c_3\ a_3 \stackrel{\mathrm{by}\ 2}{\cdots} a_{k-6})$$ is a $k$-cycle (permuting points $a_2, \ldots, a_{k-2},c_1,c_2,c_3$).

Having converted $\alpha$ to a $k$-cycle with $\leq 2$ commutations, one might hope to then use 1 conjugation to convert to the specific $k$-cycle $\beta$. However when $k = t-1$, the $k$-cycles form two distinct conjugacy classes in $A_t$ so we cannot do this. We instead note that the points permuted by the $k$-cycle depend directly on the points permuted by $\alpha$ (and the extra points $c_i$), and that products of an equal number of disjoint transpositions are conjugate in $A_t$. So by {\em first} using 1 conjugation to modify $\alpha$ appropriately, the above commutations yield $\beta$.
\end{proof}

\begin{lemma} \label{lem:evencycles}
Let $t \equiv 2\, (\bmod\ 4)$ and $\beta \in A_t$ be any product of two disjoint cycles of even lengths $k_1,k_2$. Denote $k = k_1 + k_2$. For any $\alpha \in A_t$ that is the product of either $k/2$ or $k/2 - 1$ disjoint transpositions (depending on which is even), there exist $\gamma_1, \gamma_2 \in A_t$ such that $\gamma_2^{-1} \cdot [\alpha,\gamma_1] \cdot \gamma_2 = \beta$.
\end{lemma}
\begin{proof}
We first use one commutation to convert $\alpha$ to the product of a $k_1$-cycle and a $k_2$-cycle, and then one conjugation to convert it to $\beta$ (which here we can do without the complication mentioned at the end of Lemma \ref{lem:oddcycles}). We assume that $k_1,k_2 \geq 4$, and at the end mention how to handle the two cases $k_1 = k_2 = 2$ and $k_1 = 2, k_2 = 4$.

If $k/2$ is even, let $$\alpha := (a_1\ b_1) \cdots (a_{k'_1}\ b_{k'_1})(c_1\ d_1) \cdots (c_{k'_2}\ d_{k'_2})$$ be any product of $k/2 = k_1' + k_2'$ disjoint transpositions, where $k'_1 := k_1/2$ and $k'_2 := k_2/2$. We will show that there exists $\pi \in A_t$ that is the product of $k/2$ disjoint transpositions such that $\alpha \pi$ is the product of a $k_1$-cycle and a $k_2$-cycle. As in Lemma \ref{lem:oddcycles}, this implies that we can convert $\alpha$ to the desired form by commutating with $\gamma \in A_t$ such that $\gamma \alpha^{-1} \gamma^{-1} = \pi$.
Define
$$\pi := \prod_{i=1}^{k'_1-1} (a_{i+1}\ b_i)\ \cdot\ \prod_{i=1}^{k'_2-2}(c_{i+1}\ d_i)\ \cdot\ (d_{k'_2-1}\ e_1)(c_{k'_2}\ d_{k'_2})(c_1\ e_2)$$
where $e_1,e_2 \in [t]$ are distinct from each point permuted by $\alpha$. (Such $e_1,e_2$ must exist because $4 | k$ and $t \equiv 2\, (\bmod\ 4)$, and thus $k \leq t-2$.) Then we have that
$$\alpha\pi = (a_1\ \cdots\ a_{k'_1}\ b_{k'_1}\ \cdots\ b_1)(c_1\ \cdots\ c_{k'_2-1}\ e_1\ d_{k'_2-1}\ \cdots\ d_1\ e_2)$$
is the product of two disjoint cycles of lengths $2k'_1 = k_1$ and $2k'_2 = k_2$.

If instead $k/2 - 1$ is even, let $$\alpha := (a_1\ b_1) \cdots (a_{k'_1}\ b_{k'_1})(c_1\ d_1) \cdots (c_{k'_2-1}\ d_{k'_2-1})$$ be any product of $k/2 - 1 = k'_1 + k'_2 - 1$ disjoint transpositions. This time we define
$$\pi := \prod_{i=1}^{k'_1-1} (a_{i+1}\ b_i)\ \cdot\ \prod_{i=1}^{k'_2-3}(c_{i+1}\ d_i)\ \cdot\ (c_1\ c_{k'_2-1})(d_{k'_2-2}\ e_1)(d_{k'_2-1}\ e_2)$$
where again $e_1,e_2$ are distinct from each of the $2(k/2 - 1) \leq t-2$ points permuted by $\alpha$ (here we use $k'_1 \geq 2$, $k'_2 \geq 3$). Then we have that
$$\alpha\pi = (a_1\ \cdots\ a_{k'_1}\ b_{k'_1}\ \cdots\ b_1)(c_1\ \cdots\ c_{k'_2-2}\ e_1\ d_{k'_2-2}\ \cdots\ d_1\ c_{k'_2-1}\ e_2\ d_{k'_2-1})$$
is the product of two disjoint cycles of lengths $2k'_1 = k_1$ and $2k'_2 = k_2$.

Finally we handle the two cases $k_1 = k_2 = 2$ and $k_1 = 2, k_2 = 4$. If $k_1 = k_2 = 2$ then $\alpha$ and $\beta$ are both double-transpositions and can be made equal with a single conjugation. Otherwise denote $\alpha = (a\ b)(c\ d)$, and note that there is another double-transposition $\pi = (c\ e)(d\ f)$ such that $\alpha \pi = (a\ b)(c\ f\ d\ e)$. Thus $\alpha$ can be commutated to the product of a 2-cycle and a 4-cycle, and a conjugation can make it equal to $\beta$.
\end{proof}

Lemma \ref{lem:convert} follows immediately from Lemmas \ref{lem:to-doub-transp}-\ref{lem:evencycles}. The only cases not explicitly covered by these are when $t = 2$ in which case Lemma \ref{lem:convert} is vacuous, and when $\beta$ is a 3-cycle (because Lemma \ref{lem:oddcycles} only handles $(k \geq 5)$-cycles). For the latter, note that by assumption we must have $t \geq 6$. We first convert $\alpha$ to a 5-cycle $(a_1\ \cdots\ a_5)$ using Lemmas \ref{lem:to-doub-transp}-\ref{lem:oddcycles}, then use one commutation with $(a_2\ a_3\ a_4)$ to convert to $(a_1\ a_4\ a_3)$, and finally use one conjugation to convert to $\beta$.

\section{Hardness for a single element} \label{sec:1234}
%
%
%
%

In this section we show that the $(1\ 2)(3\ 4)$-product problem is L-complete.

\begin{theorem} \label{thm:some-hard-alpha}
If for sufficiently large $t$ there is a circuit of depth $O(\log t)$ that decides the $(1\ 2)(3\ 4)$-product problem over $(A_t)^t$, then {\em NC}$^1 =$ \em{L}. 
\end{theorem}

We use the following theorem which is proved afterwards and says that deciding if an input vector has product $= \id$ is L-complete.

\begin{restatable}{theorem}{decideidhard} 
\label{thm:decide-id-hard}
If for sufficiently large $t$ there is a circuit of depth $O(\log t)$ that decides if its input in $(A_t)^t$ has product $= \id$, then {\em NC}$^1 =$ \em{L}. 
\end{restatable}

To prove Theorem \ref{thm:some-hard-alpha} from Theorem \ref{thm:decide-id-hard}, we show how to construct a set of $t^{O(1)}$ vectors from an input vector $x \in (A_t)^t$ such that, if $x$ has product $= \id$ then they all do, and otherwise some vector has product $= (1\ 2)(3\ 4)$. Then we apply the circuit deciding the $(1\ 2)(3\ 4)$-product problem to each vector, and a depth-$O(\log t)$ OR tree to the outputs. The vectors are constructed using commutation and conjugation via Lemma \ref{lem:to-doub-transp}.

\begin{proof}[Proof of Theorem \ref{thm:some-hard-alpha}]
Assume that there is a circuit $C$ of depth $O(\log t)$ that decides the $(1\ 2)(3\ 4)$-product problem. We construct a circuit $C'$ of depth $O(\log t)$ that decides if its input has product $= \id$, which in combination with Theorem \ref{thm:decide-id-hard} proves the theorem.

Lemma \ref{lem:to-doub-transp} shows that for every $\id \neq \alpha \in A_t$, there exist $\gamma_1,\gamma_2 \in A_t$ such that $\alpha' := [[\alpha,\gamma_1],\gamma_2]$ is a double-transposition and each $\gamma_i$ is either a double-transposition or a 3-cycle. We observe in the claim following this proof that for every double-transposition $\alpha'$, there exists $\gamma_3 \in A_t$ such that $\gamma_3^{-1} \cdot \alpha' \cdot \gamma_3 = (1\ 2)(3\ 4)$ and $\gamma_3$ permutes $\leq 8$ points. 

For any such choice of $\gamma := (\gamma_1,\gamma_2,\gamma_3)$, let $C_{\gamma} : (A_t)^t \to (A_t)^t$ denote a circuit of depth $O(\log t)$ that satisfies $$\quad \prod_i x_i = \alpha \quad \Longrightarrow \quad \prod_i C_\gamma(x)_i = \gamma_3^{-1} \cdot [[\alpha,\gamma_1],\gamma_2] \cdot \gamma_3$$
for every $\alpha \in A_t$ and $x \in (A_t)^t$. The crucial point is that if $\prod_i x_i \neq \id$ then there exists $\gamma$ such that $\prod_i C_\gamma (x)_i = (1\ 2)(3\ 4)$, and otherwise $\prod_i C_\gamma (x)_i = \id$ for every $\gamma$.

Observe that the number of double-transpositions in $A_t$ is ${t \choose 4} \cdot 3$, the number of $3$-cycles is ${t \choose 3} \cdot 2$, and the number of permutations that permute $\leq 8$ points is $< {t \choose 8} \cdot |A_8|$, all of which are $t^{O(1)}$. Thus on input $x$, $C'$ checks in depth $O(\log t)$ if any of these $t^{O(1)}$ choices of $\gamma = (\gamma_1,\gamma_2,\gamma_3)$ satisfies $C(C_\gamma(x)) = 1$.
\end{proof}

\begin{claim*}
Let $t \geq 8$. For every double-transposition $\alpha \in A_t$, there exists $\gamma \in A_t$ such that $\gamma^{-1} \alpha \gamma = (1\ 2)(3\ 4)$ and $\gamma$ permutes $\leq 8$ points.
\end{claim*}
\begin{proof}
Denote $\alpha = (a\ b)(c\ d)$. Any injective function $\phi : \{1,2,3,4\} \cup \{a,b,c,d\} \to [8]$ maps $\alpha$ and $(1\ 2)(3\ 4)$ to two double-transpositions in $A_8$. Since the latter are conjugate, and since the $\gamma \in A_8$ that ``witnesses'' this conjugacy necessarily permutes $\leq 8$ points, applying $\phi^{-1}$ (suitably defined) to $\gamma$ yields an element in $A_t$ that permutes $\leq 8$ points and witnesses the conjugacy of $\alpha$ and $(1\ 2)(3\ 4)$.
\end{proof}

\subsection{Proof of Theorem \ref{thm:decide-id-hard}} \label{sec:decide-id-hard}

Recall from \S 1 the following encoding of branching programs by permutations. The proof of this theorem is implicit in \cite[Prop.\ 1]{CookM87}.

\cookmckenzie*

Next we use this to show that the problem of deciding if a vector over $S_t$ has product $= \id$ is L-complete. This proof is due to Eric Allender and V.\ Arvind (personal communication), and we include it with their permission.

\begin{theorem} \label{thm:allender-arvind}
If for sufficiently large $t$ there is a circuit of depth $O(\log t)$ that decides if its input in $(S_t)^t$ has product = $\id$, then {\em NC}$^1 =$ {\em L}.
\end{theorem}
\begin{proof}
Let a string $x$ and a branching program $B$ of size $s = \poly(|x|)$ be given. Let $t = O(s)$ be as in Theorem \ref{thm:cook-mckenzie}.

We first construct $t$ vectors in $(S_{t})^{t}$ such that $B$ accepts $x$ iff the product of some vector is a permutation that maps $1 \mapsto t$. 
Let $\sigma \in S_{t}$ be given by Theorem \ref{thm:cook-mckenzie}. Notice that $1$ and $t$ are in the same cycle in $\sigma$'s disjoint cycle representation iff $\exists k \leq t$ such that $\sigma^k$ maps $1 \mapsto t$. Thus for $k = 1,\ldots,t$ we construct the $k$th vector to have product $\sigma^k$, by concatenating $k$ copies of $\sigma$ and $t-k$ copies of $\id$.
(Up to now this construction appears in \cite{CookM87}.)

Next we transform $z \in (S_{t})^{t}$ to $z' \in (S_{t+1})^{2t+2}$ satisfying 
\begin{equation} \label{eq:map-to-id}
\prod_i z_i \mbox{ maps }1 \mapsto t \quad \Longleftrightarrow \quad \prod_i z'_i = \id.
\end{equation}
This is done via the map $$z' := (z,\ (t\ \ t+1),\ z^{-1},\ (1\ \ t+1))$$ where $z^{-1} := (z^{-1}_{t},\ldots,z^{-1}_1) \in (S_{t})^{t}$ and we embed $S_{t}$ into $S_{t+1}$ in the canonical way. This is computable in depth $O(\log t)$.

To see that (\ref{eq:map-to-id}) holds, denote $\pi := \prod_i z_i$ and $\pi' := \prod_i z'_i$. If $\pi(1) \neq t$, then $\pi'(1) = t+1$ and so $\pi' \neq \id$. On the other hand if $\pi(1) = t$, then it can be checked that $\pi'(1) = 1$ and $\pi'(t+1) = t+1$, and it is clear that $\pi'(j) = j$ for $1 < j < t+1$ since $\pi(j)$ is not touched by $(t\ \ t+1)$ and $j$ is not touched by $(1\ \ t+1)$.

Thus we have constructed $t$ vectors in $(S_{t+1})^{2t+2}$ such that $B$ accepts $x$ iff some vector has product $=\id$. We can reduce the vectors' length to $t+1$ in depth $O(\log t)$ by multiplying adjacent permutations. Finally, applying the circuit in the assumption of the theorem and an OR-tree yields a circuit of depth $O(\log t)$ that decides if $B$ accepts $x$. 
\end{proof}

To conclude the proof of Theorem \ref{thm:decide-id-hard}, we observe that there is an embedding $M : S_t \to A_{t+2}$ computable in NC$^1$ that preserves the identity product. $M$ is defined by $M(\alpha) := \alpha$ if $\alpha$ is even and $M(\alpha) := \alpha \cdot (t+1\ \ t+2)$ if $\alpha$ is odd. It can be checked that this is a homomorphism, and thus $\prod_i x_i = \id \Leftrightarrow \prod_i M(x_i) = \id$.
Computing $M$ requires deciding if $\alpha \in S_t$ is odd or even, which can be done in depth $O(\log t)$ by checking if there are an odd or even number of pairs $i < j \leq t$ such that $\alpha(i) > j$.

\section{Proof of Corollary \ref{thm:nc1}} \label{sec:puttingtogether}

Theorem \ref{thm:allalphahard} is immediate from Theorems \ref{thm:localmap} and \ref{thm:some-hard-alpha}. We now prove Corollary \ref{thm:nc1} from Theorem \ref{thm:allalphahard} using the random self-reducibility of group products (cf.\ \cite[Thm.\ 3.9]{MilesV-leak}).

\thmallalphahard*

\thmmain*

\begin{proof}[Proof of Corollary \ref{thm:nc1}]
Assume that there exists $k$ such that for sufficiently large $t$, there exists $\alpha \in A_t$ and a circuit $C$ of depth $k \log t$ such that $\Delta(C(D_\alpha),C(D_\id)) \geq t^{-k}$, where recall that $D_g$ is the uniform distribution over $\{x \in (A_t)^t\, |\, \prod_i x_i = g\}$. 

Let $S \subseteq \zo^{k \log t}$ be the set that maximizes $\Pr[C(D_\alpha) \in S] - \Pr[C(D_\id) \in S]$, and note that checking $x \in S$ can be done in depth $O(\log t)$. Thus, there exists $C' : (A_t)^t \to \zo$ of depth $O(\log t)$ such that
\begin{equation}
\label{eqn-circuit}
\Pr[C'(D_\alpha) = 1] - \Pr[C'(D_\id) = 1] \geq t^{-k}.
\end{equation}
Define $\eps_\alpha := \Pr[C'(D_\alpha) = 1]$ and $\eps_\id := \Pr[C'(D_\id) = 1]$, and note that $\eps_\alpha \geq t^{-k}$.

Let $C''$ be the randomized circuit of depth $O(\log t)$ that computes as follows on input $x \in (A_t)^t$. First for $m := t^{3k+3}/\eps_\alpha = t^{O(k)}$, $C''$ samples $z_1,\ldots,z_m \in (A_t)^t$ independently from $D_g$ where $g := \prod_i x_i \in \{\alpha,\id\}$. This is done by choosing uniform $r_1,\ldots,r_{t-1} \in A_t$ for each $i \leq m$ and computing $z_i := (x_1 r_1,\ r_1^{-1} x_2 r_2,\ \ldots,\ r_{t-1}^{-1} x_t)$.
Then $C''$ outputs $\alpha$ if $$(1-1/(2t^k)) \cdot m\eps_\alpha\ \leq\ \displaystyle\sum_{i=1}^m C'(z_i)\ \leq\ (1+1/(2t^k)) \cdot m\eps_\alpha$$ and otherwise outputs $\id$.

We now prove the following claim.
\begin{claim*}
For all $x \in (A_t)^t$ such that $\prod_i x_i \in \{\alpha,\id\}$, we have $\Pr[C''(x) = \prod_i x_i ] > 1-|A_t|^{-t}$ over the random coins of $C''$.
\end{claim*}
This implies the theorem, as follows. By a union bound there is a way to fix the random coins of $C''$ such that $C''(x) = \prod_i x_i$ for every $x$ satsifying $\prod_i x_i \in \{\alpha,\id\}$. This solves the $\alpha$-product problem in NC$^1$, and thus by Theorem \ref{thm:allalphahard} we have NC$^1 =$ L.
\begin{proof}[Proof of Claim.]
Denote $X := \sum_{i=1}^m C'(z_i)$, and $\mu := \mathbb{E}[X]$. Note that $|A_t|^t = (t!/2)^t < 2^{t^3}$.

Fix $x$, and first assume $\prod_i x_i = \alpha$ which means $\mu = m\eps_\alpha = t^{3k+3}$. Then $$\Pr[C''(x) = \alpha] = \Pr\left[|X - \mu| \leq \mu /(2t^k) \right] \geq 1 - 2e^{-\mu \cdot t^{-3k}} \geq 1 - 2^{-t^3}$$ by a Chernoff bound.

Now assume $\prod_i x_i = \id$. Then $\mu = m\eps_\id$, and since $\eps_\alpha/\eps_\id \geq 1 + 1/t^k$ by (\ref{eqn-circuit}), we have $(1-1/(2t^k)) \cdot m\eps_\alpha \geq \mu(1 + 1/(3t^k))$.
Then using another Chernoff bound, we have
$$\Pr[C''(x) = \id] \geq 1 - \Pr[X \geq \mu(1 + 1/(3t^k))] \geq 1 - e^{-\mu \cdot t^{-3k}} \geq 1 - 2^{-t^3}. \qedhere$$
\end{proof}
This completes the proof of the theorem.
\end{proof}

\paragraph{Acknowledgements.} 
We are grateful to Eric Allender and V.\ Arvind for sharing the proof of Theorem \ref{thm:allender-arvind} and allowing us to include it here, and to Emanuele Viola for helpful discussions and comments on a previous draft of this paper.

\bibliographystyle{alpha}
\newcommand{\etalchar}[1]{$^{#1}$}
\def\cprime{$'$}

\end{document}